\documentclass[a4paper,USenglish]{lipics-v2018}
\usepackage{microtype} 
\title{Parameterized Leaf Power Recognition via Embedding into Graph Products}
\titlerunning{Parameterized Leaf Power Recognition}

\author{David Eppstein}{Computer Science Department, University of California, Irvine, USA}{eppstein@uci.edu}{}{Supported in part by NSF grants  CCF-1618301 and CCF-1616248.}
\author{Elham Havvaei}{Computer Science Department, University of California, Irvine, USA}{ehavvaei@uci.edu}{}{}
\authorrunning{D. Eppstein and E. Havvaei}

\Copyright{David Eppstein and Elham Havvaei}

\subjclass{ }

\keywords{leaf power, phylogenetic tree, monadic second-order logic, Courcelle's theorem, strong product of graphs, fixed-parameter tractability, dynamic programming, tree decomposition }

\nolinenumbers 
\hideLIPIcs  

\usepackage{aliascnt}
\makeatletter
\let\corollary\@undefined
\let\endcorollary\@undefined
\let\lemma\@undefined
\let\endlemma\@undefined
\makeatother
\theoremstyle{theorem} 
\newaliascnt{lemma}{theorem}
\newtheorem{lemma}[lemma]{Lemma}
\aliascntresetthe{lemma}

\newaliascnt{corollary}{theorem}
\newtheorem{corollary}[corollary]{Corollary}
\aliascntresetthe{corollary}

\newcommand{\edge}{\textsc{edge}}
\newcommand{\nonedge}{\textsc{nonedge}}

\newcommand{\MSO}{\mathrm{MSO}} 
\newcommand{\incident}{\multimap} 
\newcommand{\adjacent}{\textsc{adjacent}} 
\newcommand{\leaf}{\textsc{leaf}} 
\newcommand{\horizontal}{\textsc{horizontal}} 
\newcommand{\acyclic}{\textsc{acyclic}} 
\newcommand{\alignedwith}{\textsc{aligned}} 
\newcommand{\representative}{\textsc{representative}} 
\newcommand{\represented}{\textsc{represented}} 
\newcommand{\haspath}{\textsc{path}} 
\newcommand{\isroot}{\textsc{root}} 

\hideLIPIcs
\EventEditors{Christophe Paul and Micha\l{} Pilipczuk}
\EventNoEds{2}
\EventLongTitle{13th International Symposium on Parameterized and Exact Computation (IPEC 2018)}
\EventShortTitle{IPEC 2018}
\EventAcronym{IPEC}
\EventYear{2018}
\EventDate{August 20--24, 2018}
\EventLocation{Helsinki, Finland}
\EventLogo{}
\SeriesVolume{115}
\ArticleNo{16} 

\begin{document}

\maketitle

\begin{abstract}
The $k$-leaf power graph $G$ of a tree $T$ is a graph whose vertices are the leaves of $T$ and whose edges connect pairs of leaves at unweighted distance at most~$k$ in $T$. Recognition of the $k$-leaf power graphs for $k \geq 7$ is still an open problem. In this paper, we provide two algorithms for this problem for sparse leaf power graphs. Our results shows that the problem of recognizing these graphs is fixed-parameter tractable when parameterized both by $k$ and by the degeneracy of the given graph. To prove this, we first describe how to embed a leaf root of a leaf power graph into a product of the graph with a cycle graph. We bound the treewidth of the resulting product in terms of $k$ and the degeneracy of $G$. The first presented algorithm uses methods based on monadic second-order logic ($\MSO_2$) to recognize the existence of a leaf power as a subgraph of the graph product. Using the same embedding in the graph product, the second algorithm presents a dynamic programming approach to solve the problem and provide a better dependence on the parameters. 


\end{abstract}

\section{Introduction}
\label{sec:intro}

\emph{Leaf powers} are a class of graphs that were introduced in 2002 by Nishimura, Ragde and Thilikos \cite{MR1874637}, extending the notion of graph powers. For a graph $G$, the $k$th power
graph $G^k$ has the same set of vertices as $G$ but a different notion of adjacency: two vertices are adjacent in $G^k$ if there is a path of at most $k$ edges between them in $G$. Determining whether a graph is a $k$th power of another graph is known to be NP-complete, for $k \geq 2$~\cite{nguyen2009hardness}. However deciding whether a graph $G$ is the second power of a graph $H$ is decidable in polynomial time when $H$ belongs to various graph classes such as bipartite graphs~\cite{lau2006bipartite}, block graphs~\cite{tuy2010square}, cactus graphs~\cite{golovach2016finding} and cactus block graphs~\cite{ducoffe2019finding}. Besides, it is possible to decide in linear time if a graph is the power of a tree~\cite{chang2015linear}. The leaf powers are defined in the same way from trees, but only including the leaves of the trees as vertices. The $k$th leaf power of a tree $T$ has the leaves of $T$ as its vertices,
with two vertices adjacent in the leaf power if there is a path of at most $k$ edges between them in $T$. A given graph $G$ is a $k$-leaf-power graph when there exists a tree $T$ for which $G$ is the $k$th leaf power. In this case, $T$ is a \emph{$k$-leaf root} of $G$. In general, the $k$-leaf root may have vertices and edges that are not part of the input graph. For example, \autoref{fig:3-leaf} shows a 3-leaf power alongside one of its 3-leaf roots. Nishimura et al., further, derived the first polynomial-time algorithms to recognize $k$-leaf powers for $k = 3$ and $k = 4$~\cite{MR1874637}.  

One application of recognizing leaf powers arises as a formalization of a problem in computational biology, the reconstruction of evolutionary history and evolutionary trees from information about the similarity between species \cite{MR2001887,fitch1967construction}. In this problem, the common ancestry of different species can be represented by an evolutionary or phylogenetic tree, in which each vertex represents a species and each edge represents a direct ancestry relation between two species. We only have full access to living species, the species at the leaves of the tree; the other species in the tree are typically long-extinct, and may be represented physically only through fossils or not at all. If we suppose that we can infer, from observations of living species, which ones are close together (within some number $k$ of steps in this tree) and which others are not, then we could use an algorithm for leaf power recognition to infer a phylogenetic tree consistent with this data.
  
 \begin{figure}[t]
   \center
      \includegraphics[scale=.33]{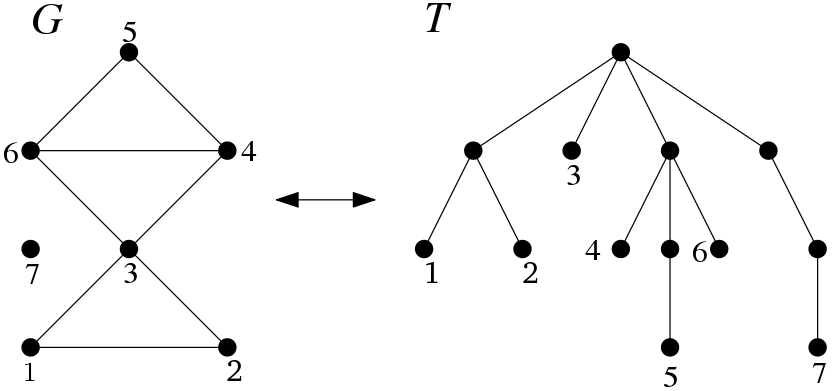}      
\caption{A 3-leaf power graph $G$ and one of its 3-leaf roots $T.$  }
\label{fig:3-leaf}
\end{figure}  

\subsection{New Results}
In this paper, presenting two different algorithms, we prove that the $k$-leaf powers of degeneracy $d$ can be recognized in time that is fixed-parameter tractable when parameterized by $k$ and $d$. Here, the degeneracy of a graph is the maximum, over its subgraphs, of the minimum degree of any subgraph. 

Our first algorithm makes ample use of Courcelle's theorem~\cite{MR1217156} while the second employs a dynamic programming method to provide a time complexity with a better dependence on the parameters. Although the second algorithm is more efficient, we retain the description of the first algorithm as it was the source of our inspiration to devise a more practical method to prove the fixed-parameter tractability of $k$-leaf powers, and as we feel that our technique of using graph products
(which we use in both algorithms) can have broader applications.

Both algorithms have running time polynomial (in fact linear) in the size of the input graph, multiplied by a factor that depends non-polynomially on $k$ and $d$. We also apply the same methods to a more general problem in which each edge of the input graph is labeled by a range of distances, constraining the corresponding pair of leaves in the leaf root to have a distance in that range.

Later, it will be discussed that leaf powers have unbounded clique-width. However, it is known that the $k$-leaf powers have bounded clique-width when $k$ is bounded~\cite{gurski2007clique}.
A wide class of graph problems (those expressible in a version of monadic second order logic quantifying over only vertex sets, $\MSO_1$) can be solved in fixed-parameter time for graphs of bounded clique-width, via Courcelle's theorem. However we have been unable to express the recognition of leaf powers in $\MSO_1$. Instead, our algorithm uses a more powerful version of monadic second order logic allowing quantification over edge sets, $\MSO_2$. Later, it will be discussed that leaf powers with bounded degeneracy have bounded treewidth, allowing us to apply a form of Courcelle's theorem for $\MSO_2$ for graphs of bounded treewidth.

However, there is an additional complication that makes it tricky to apply these methods to leaf power recognition.
As stated earlier, the tree that we wish to find, for which our given input graph is a leaf power, will in general include vertices and edges that are not part of the input, but $\MSO_2$ can only quantify over subsets of the existing vertices and edges of a graph, not over sets of vertices and edges that are not subsets of the input. To work around this problem,
we apply Courcelle's theorem not to the given graph $G$ itself, but to a \emph{graph product} $G \boxtimes C_k$
where $C_k$ is a $k$-vertex cycle graph. We prove that a leaf root (the tree for which $G$ is a leaf power, if there is one) can be embedded as a subgraph of this product, that it can be recognized by an $\MSO_2$ formula applied to this product, and that this product has bounded treewidth whenever $G$ is a $k$-leaf power of bounded degeneracy.
In this way we can recognize $G$ as a leaf power, not by applying Courcelle's theorem to $G$, but by applying it to the graph product.

Thus, our algorithm combines the following ingredients:
\begin{itemize}
\item Our embedding of the $k$-leaf root as a subgraph of the graph product $G \boxtimes C_k$.
\item Our logical representation of $k$-leaf roots as subgraphs of graph products.
\item Courcelle's theorem, which provides general-purpose algorithms for testing $\MSO_2$ formulas on graphs of bounded treewidth.
\item The fact that leaf powers of bounded degeneracy also have bounded treewi-dth.
\item The fact that, by taking a product with a graph of bounded size, we preserve the bounded treewidth of the product.
\end{itemize}
Our algorithm runs in fixed-parameter tractable time when parameterized by $k$ and the degeneracy $d$ of the given input graph. In particular, it runs in linear-time when $k$ and $d$ are both constant.

Our results provide the first known efficient algorithms for recognizing $k$-leaf powers for $k \ge 7$, for graphs of bounded degeneracy.
More generally, our method of embedding into graph products appears likely to apply to other graph problems involving network design (the addition of edges to an existing graph, rather than the identification of a special subgraph of the input). In the case we apply this method to leaf power recognition, we expect that it should be possible to translate
our $\MSO_2$ formula over the graph product into a significantly more complicated $\MSO_2$ formula over the input graph, but the method of embedding into graph products considerably simplifies our task of designing a logical formula for our problem. Later, we also profit from the same embedding into a product as a key step in our dynamic programming algorithm to decide whether a graph is a $k$-leaf power.

\subsection{Related Work}
\begin{figure}[t]
   
    \includegraphics[width=\linewidth]{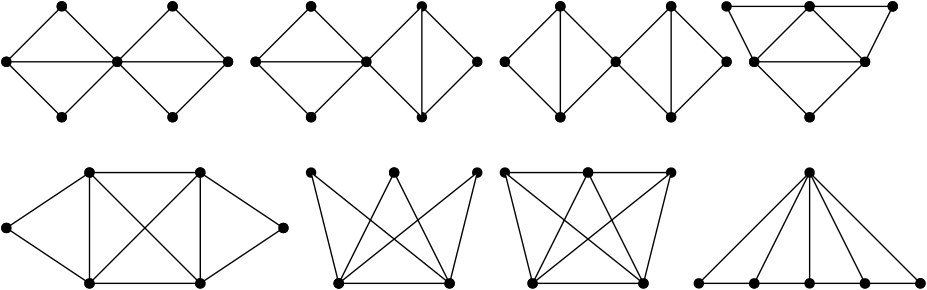}\par 
     
\caption{ A graph is a 4-leaf power if and only if it is chordal and does not contain any of the graphs above as a subgraph. }
\label{fig:forbidden}
\end{figure}
Polynomial-time algorithms are known for recognizing $k$-leaf powers for $k \leq 6$.
\begin{itemize}
\item A graph is a 2-leaf power if it is a disjoint union of cliques, so this class of graphs is trivial to recognize.
\item There exist various ways to characterize 3-leaf powers~\cite{MR1874637,MR2211095,dom2004error,MR2237730}, some of which lead to efficient algorithms. For instance, one way to determine if a graph is a 3-leaf power is to check whether it is bull-, dart- and gem-free and chordal~\cite{dom2004error}. The chordal graphs have a known recognition algorithm, and testing for the existence of any of the other forbidden induced subgraphs is polynomial, because they all have bounded size.
\item Similarly, there are various known ways to characterize  4-leaf powers~\cite{MR1874637,MR2237730,dom2005extending,MR2479182}. One is that a graph is a 4-leaf power if and only if it is chordal and does not contain any of the graphs depicted in \autoref{fig:forbidden} as induced subgraphs~\cite{MR2237730}. Again, this leads to a polynomial-time recognition algorithm, because all of these graphs have bounded size.
\item $k$-leaf powers can be recognized in polynomial time if the $(k-2)$- Steiner root problem can be solved in polynomial time. Chang and Ko, in 2007, provided a linear-time recognition algorithm for 3-Steiner root problem~\cite{chang20073}. This implies that 5-leaf powers can be recognized in linear time. Besides, Brandst{\"a}dt et al. provided a forbidden induced subgraph characterization for the distance-hereditary 5-leaf powers~\cite{MR2537378}.

\item Ducoffe has recently extended result of Chang and Ko~\cite{chang20073} and provided a polynomial-time recognition algorithm of 4-Steiner powers~\cite{ducoffe20194} which as stated, it leads to a polynomial-time recognition of 6-leaf powers.
\end{itemize}

Polynomial-time structural characterization of $k$-leaf powers for  $k \geq 7$ is still an open problem. 
  
Throughout the literature, there exist many structural characterizations of leaf powers which provide potentially useful insight into this class of graphs. It is known, for instance, that all leaf powers are strongly chordal, but the converse is not always true. Further, Kennedy et al. showed that strictly chordal graphs are always $k$-leaf powers for $k \geq 4$; these are the chordal graphs that are also, dart- and gem-free. They provided a linear-time algorithm to construct $k$-leaf roots of strictly chordal graphs \cite{MR2577678}.

For all $k \geq 2$, every $k$-leaf power is also a ($k+2$)-leaf power. A $(k+2)$-leaf root of any $k$-leaf-power can be obtained from its $k$-leaf root, by subdividing all edges incident to leaves. However, the problems of recognizing $k$-leaf powers for different values of $k$ do not collapse: for all $k \geq 4$, there exists a $k$-leaf power which is not a $(k+1)$-leaf power  \cite{brandstadt2008k}. 

\subsection{Organization}

This paper is organized as follows.
We begin in \autoref{sec:preliminaries} with some preliminary definitions and a survey of the relevant background material for our results.
In \autoref{sec:embedding} we describe how to embed leaf roots into graph products , a construction used in both of our algorithms.
We provide a logical formulation of the leaf power recognition problem in \autoref{sec:logic}, and in \autoref{sec:courcelle} we use this formulation for our first algorithm for the problem.
We generalize the problem to leaf powers with restricted distance ranges on each input graph edge in \autoref{sec:generalize}.
Our dynamic programming algorithm for leaf powers is presented in \autoref{sec:dynamic}.
We conclude with some general observations in \autoref{sec:conclusion}.

\section{Preliminaries}
\label{sec:preliminaries}

\subsection{Definitions}
Throughout this paper, we let $G(V,E)$ denote a simple undirected graph (typically, the input to the leaf power recognition problem). If $u$ and $v$ are two vertices in $V$ that are adjacent in $G$, we let $e(u,v)$ denote the edge connecting them. 
  \begin{figure}[t!]
   \center
      \includegraphics[scale=.35]{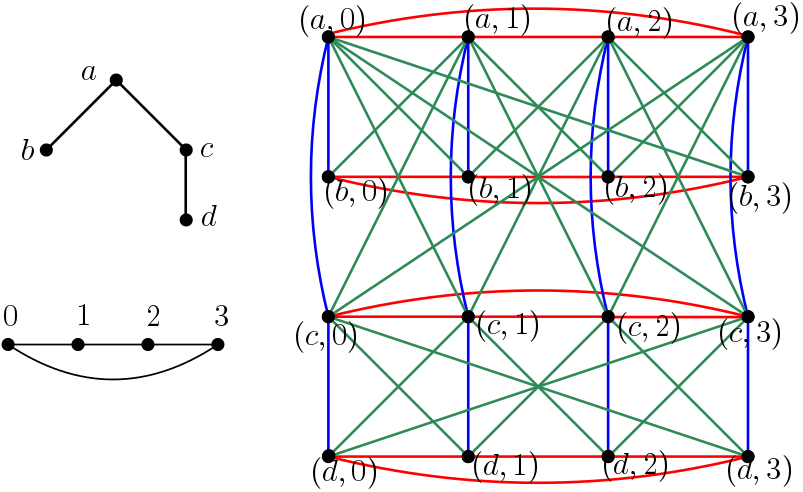}      
\caption{ The graph on the right is the strong product of a four-vertex path graph (top left) and a four-vertex cycle graph (bottom left). The colors indicate the partition of the edges into vertical, horizontal, and diagonal subsets.}
\label{cross}
\end{figure}

The strong product of graphs $G_1$ and $G_2$, denoted as $G_1  \boxtimes G_2$, is a graph whose vertices are ordered pairs of a vertex from $G_1$ and a vertex from $G_2$. In it, two distinct vertices $(u_1,u_2)$ and $(v_1,v_2)$ are adjacent if and only if for all $i \in \{1,2\}$,  $u_i = v_i$ or $u_i$ and $v_i$ are adjacent in $G_i$. \autoref{cross} shows an example, the strong product of a four-vertex path graph with a four-vertex cycle graph. When we construct a strong product, we will classify the edges of the product into three subsets:
\begin{itemize}
\item We call an edge from $(u_1,u_2)$ to $(v_1,v_2)$ a \emph{vertical edge} if $u_2=v_2$.
The edges of this type form $|V(G_2)|$ disjoint copies of $G_1$ as subgraphs of the product.
\item We call an edge from $(u_1,u_2)$ to $(v_1,v_2)$ a \emph{horizontal edge} if $u_1=v_1$.
The edges of this type form $|V(G_1)|$ disjoint copies of $G_2$ as subgraphs of the product.
\item We call the remaining edges, for which $u_1\ne v_1$ and $u_2\ne v_2$, \emph{diagonal edges}. The subgraph composed of the diagonal edges forms a different kind of graph product,
the \emph{tensor product} $G_1\times G_2$.
\end{itemize}
We may think of these three edge sets as forming an (improper) edge coloring of the graph product.
In \autoref{cross} these edge sets are colored  blue, red and green, respectively. 
\subsection{Graph Parameters}
\label{subsec:parameters}
One of the simplest ways of parameterizing sparse graphs is by their \emph{degeneracy}. The degeneracy $d(G)$ of a graph $G$ is the smallest number such that every nonempty subgraph of $G$ contains at least one vertex of degree at most $d(G)$~\cite{lick1970k}. Degeneracy may be equivalently defined as the least $d$ for which an ordering of vertices of the graph exists in which each vertex has at most $d$ later neighbors in that ordering. There exists many problems in literature parameterized by the degeneracy as a measure of graph sparseness~\cite{eppstein2010listing,alon2009linear,cai2006random}, as it implies that every graph of size $n$ and degeneracy $d$ has at most $(n-1)d$ edges. Degeneracy may be computed in linear-time by a greedy algorithm that repeatedly removes the minimum-degree vertex and records the largest degree seen among the vertices at the time they are removed~\cite{MR709826}.

The notion of \emph{treewidth}, a more complicated graph sparsity parameter, was first introduced by Bertel\'e and Brioschi~\cite{BB72} and Halin~\cite{MR0444522} and later rediscovered  by Robertson and Seymour \cite{MR855559}. One way to define treewidth is to use the concept of tree decomposition. A tree decomposition of  graph $G$ consists of a tree $T$ where each vertex $X_i \in T$ (called a bag) is a subset of vertices of $G$.
This tree and its bags are required to satisfy the following properties:
\begin{itemize}
\item For each edge $e(u,v)$ in $G$, there exists a bag in $T$ containing both $u$ and $v$; and
\item For each vertex $v$ in $G$, the bags containing $v$ form a nonempty connected subtree of $T$.
\end{itemize}
The width of a tree decomposition is the size of its largest bag, minus one. The treewidth of a graph is defined as the minimum width achieved over all tree decompositions of the graph. Bounded treewidth graphs are especially interesting from an algorithmic point of view. Many well-known NP-complete problems have linear-time algorithms on graphs of bounded treewidth~\cite{MR1268488}. 

Another related graph parameter, \emph{ clique-width}, was introduced by Courcelle et al.{} to characterize the structural complexity of graphs \cite{MR1217156}. The clique-width of a graph $G$ is the minimum number of labels necessary to construct $G$ by means of four graph operations: creation of a new vertex with a label, vertex disjoint union of labeled graphs, insertion of an edge between two vertices with specified labels and relabeling of vertices. Relevantly for us, Courcelle et al. showed that unit interval graphs are of unbounded clique-width. A graph is an interval graph if and only if all its vertices can be mapped into intervals on a straight line such that two vertices are adjacent when the corresponding intervals intersect each other. In the unit interval graphs, each interval has a unit length. As shown by Brandst{\"a}dt et al., unit interval graphs belong to the class of leaf powers, which implies that leaf powers also have unbounded clique-width \cite{brandstadt2008ptolemaic,MR2574841}. 

These three properties are defined differently to each other, and may have significantly different values. For instance, the complete bipartite graph $K_{n,n}$ has clique-width two but treewidth and degeneracy $n$, and the $n\times n$ grid graph has degeneracy two but clique-width and treewidth $\Omega(n)$. Nevertheless, as stated earlier, all leaf powers are chordal graphs and it is known for a chordal graph, treewidth is equal to maximum clique number minus one~\cite{MR855559}. This implies that treewidth of leaf powers are equal to their degeneracy. 
\subsection{Courcelle's Theorem}
By considering graphs as logical structures, their properties can be expressed in first-order and second-order logic. In first-order logic, graph properties are expressed as logical formulas wherein the variables range over vertices and the predicates include equality and adjacency relations. Second-order logic is an extension of first-order logic with the power to quantify over relations.  Particularly, many natural graph properties can be described in monadic second-order logic, which is a restriction of second-order logic in which only unary relations (sets of vertices or edges) are allowed~\cite{MR1480957}.

There exist two variations of monadic second-order logic: $\MSO_1$ and $\MSO_2$. In $\MSO_1$, quantification is allowed only over sets of vertices, while $\MSO_2$ allows quantification over both sets of vertices and sets of edges. $\MSO_2$ is strictly more expressive; there are some properties, such as Hamiltonicity \cite{MR1451381}, which are expressible in $\MSO_2$ but not in $\MSO_1$. A graph property is \emph{$\MSO_2$-expressible} if there exists an $\MSO_2$ formula to express it, in which case the corresponding class of graphs becomes \emph{$\MSO_2$-definable}.

The algorithmic connection between treewidth and monadic second-order logic is given by Courcelle's theorem, according to which every property definable in monadic second-order logic can be tested in linear time on graphs of bounded treewidth \cite{MR1042649}. Later, Courcelle et al. extended this theorem to the class of graphs with bounded clique-width when the underlying property is $\MSO_1$-definable~\cite{MR1739644}. In our application of Courcelle's theorem, we will use an $\MSO_2$ formula with a free variable $\horizontal$, an edge set, which we will use to pass to the formula certain information about the structural decomposition of the graph it is operating on. This extension of Courcelle's theorem to formulas with a constant number of additional free variables, whose values are assigned through some extra-logical process prior to applying the theorem, is non-problematic and standard.

However, even in $\MSO_2$, it is only possible to quantify over subsets of vertices and edges that belong to the graph to which the logical formula is applied. Much of the difficulty of the leaf power problem rests in this restriction. If we could quantify over edges and vertices that were not already present,
we could construct a formula that asserts the existence of sets of vertices and edges forming a leaf root of a given graph, and then add clauses to the formula that ensure that the quantified sets describe a valid leaf root. However, we are not allowed such quantification, because in general the leaf root has vertices and edges that do not belong to our input graph. To apply Courcelle's theorem to leaf power recognition, we must instead find a way to express the property of being a leaf power using only quantification over subsets of vertices and edges of the graph to which we apply the theorem. For this reason, the problem of leaf power recognition forms an important test case for the ability to express graph problems in MSO logic.

A problem is \emph{fixed-parameter tractable} with respect to a parameter $x$ of the input if the problem can be solved in time $f(x)n^{O(1)}$ where $n$ is the size of the input, $f$ is a computable function of $x$ (independent of $n$), and the exponent of $n$ in the $O(1)$ term is independent of~$x$. Courcelle's theorem is the foundation of many fixed-parameter tractable algorithms \cite{bannister2014crossing,MR2120320,MR3761167,MR2220663}, as it proves that properties expressible in $\MSO_1$ or $\MSO_2$ are fixed-parameter tractable with respect to the clique-width or treewidth (respectively) of the input graph.

\section{Embedding Leaf Roots into Graph Products}
\label{sec:embedding}

In this section, we show that every $k$-leaf power has a $k$-leaf root that can be embedded in the graph product $G \boxtimes C_k$. 
Let $G$ be a $k$-leaf power graph, and $T$ be a $k$-leaf root of $G$.
If $G$ is not connected, we can handle each of its connected components independently;
in this way, we can assume from now on, without loss of generality that $G$ is a connected graph with at least three vertices, and that $T$ is a leaf root chosen arbitrarily among the possible $k$-leaf roots of $T$. It follows from these assumptions that $T$ is a tree, because every edge in $G$ must be represented by a path in $T$. Because $T$ has at least three leaves, it has at least one interior node; we choose one of these nodes arbitrarily to be the root of $T$.
Additionally, every vertex or edge of $T$ participates in a path of length at most $k$ between two leaves, representing an edge of $G$. For, if some vertices and edge do not participate in these paths, removing all non-participating vertices and edges from $T$ would produce a smaller leaf root, without creating any new leaves. But this removal would disconnect pairs of leaves on the opposite sides of any removed edge, contradicting the assumption that $G$ is connected.
\begin{figure}[t!]
\center
\includegraphics[scale=.25]{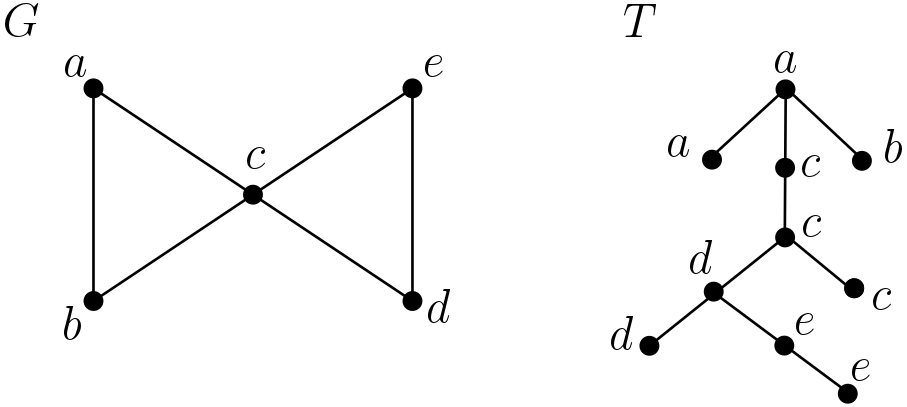}      
\caption{A 4-leaf power graph $G$ (left), and one of its leaf roots $T$ (right). Each leaf of $T$ is labeled by the vertex of $G$ that it represents, and each internal node of $T$ is labeled by its closest leaf node.
When there are ties at a node (as for instance at the root of $T$) the choice of label is made arbitrarily among the closest leaf nodes whose labels appear among the children of the node.}
\label{4_leaf_rep}
\end{figure}  

As the first step of the embedding, we provide a subroutine that takes as input, a graph and a $k$-leaf root of the form, mentioned above and embeds it in $G \boxtimes C_k$ as a subgraph. While, our leaf-power recognition algorithm does not employ this subroutine, as it does not have access to the $k$-leaf root; this subroutine solely fulfills the purpose of proving that the $k$-leaf root of this form can be embedded in the graph product. For that, we label the vertices of $T$ with the names of vertices in $G$.
Each vertex of $T$ will get a label in this way; some labels will be used more than once.
In particular, we label each leaf of $T$ by the vertex of $G$ represented by that leaf.
Then, as shown in \autoref{4_leaf_rep}, we give each non-leaf node of $T$ the same label as its closest leaf. If there are two or more closest leaves, we choose one arbitrarily among the labels already applied to the children of the given interior node. In this way, when the same label appears more than once, the tree nodes having that label form a connected path in~$T$.

As we now show, these labels, together with the depths of the nodes modulo~$k$, can be used to embed the $k$-leaf root $T$ into the strong product $G\boxtimes C_k$,
where $C_k$ denotes a $k$-vertex cycle graph.
\begin{lemma}
\label{lem:leaf-root-embedding}
If $G$ is a connected $k$-leaf power graph on three or more vertices, and $T$ is any $k$-leaf root of $G$, then $T$
can be embedded as a subtree of the strong product $G  \boxtimes C_k$.  
Additionally, the embedding can be chosen in such a way that each horizontal cycle in the strong product (the product of a vertex $v$ of $G$ with $C_k$) contains exactly one leaf of the embedded copy of $T$, the leaf representing $v$.
\end{lemma}
\begin{figure}[t]
\center
\includegraphics[scale=.23]{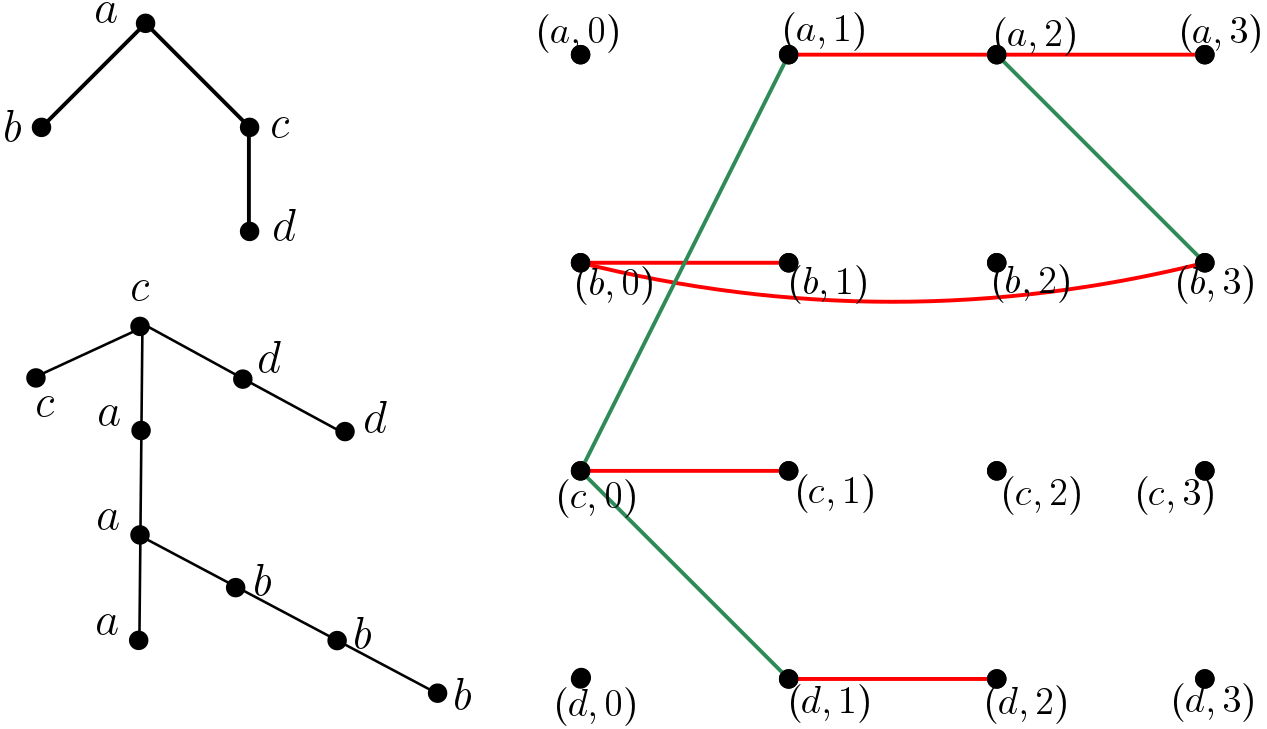}      
\caption{ The graph on the bottom left is a 4-leaf root $T$ of graph $G$ (top left).  $T$ can be embedded in the strong product $G \boxtimes C_4$ (right), by mapping each vertex $u$ of $T$ to the pair $(v,i)$ where $v$ is the label of $u$ and $i$ is the depth of $u$ (modulo~$k$).}
\label{embed}
\end{figure}  
\begin{proof}
We map a vertex $u$ of $T$ to the pair $(v,i)$
where $v$ is the label assigned to $u$ (the name of a vertex in $G$)
and $i$ is the depth of $u$ (its distance from the root of $T$), taken modulo~$k$.
This pair is one of the vertices of the strong product, so we have mapped vertices of $T$ into vertices of the strong product. An example of such embedding can be seen in  \autoref{embed}. Because $G$ is assumed to be connected, each node of $T$ participates in at least one path of length at most $k$ between two leaves of $T$, representing an adjacency of $G$; it follows that  the label for each node of $T$ is at most $k-1$ steps away from the node, and that each path of same-labeled nodes in $T$ has length at most $k-1$. As a consequence, when we take depths modulo~$k$, none of these paths can wrap around the cycle and cover the same vertex of the graph product more than once. That is, our mapping from $T$ to $G  \boxtimes C_k$ is one-to-one.
Because each leaf of $T$ is labeled with the vertex of $G$ that it represents, this mapping has the property described in the lemma, that each horizontal cycle in the strong product contains exactly one leaf of the embedded copy of $T$, the leaf representing the vertex whose product with $C_k$ forms that particular horizontal cycle.

We must also show that this mapping from $T$ to $G \boxtimes C_k$
maps each pair of vertices that are adjacent in $T$ into a pair of vertices that are adjacent in $G \boxtimes C_k$. Recall that adjacency in $G \boxtimes C_k$ is the conjunction of two conditions:
two vertices in the product are adjacent if their first coordinates are equal or adjacent in $G$ and their second coordinates are equal or adjacent in $C_k$.
Because every two adjacent vertices in $T$ have depths that differ by one,
the second coordinates of their images in the product will always be adjacent in $C_k$.
It remains to show that, when two vertices are adjacent in $T$, their images in the product have first coordinates that are equal or adjacent in $G$. That is, the labels of the two adjacent vertices in $T$ should be equal or adjacent.

Rephrasing what we still need to show, it is the following: whenever two adjacent vertices in $T$ have different labels, those labels represent adjacent vertices in $G$.

To see that this is true, consider two adjacent vertices $u_1$ and its parent $u_2$ in $T$,
labeled by two different vertices $v_1$ and $v_2$ in $G$.
As we already stated at the start of this section, the assumption of the lemma that $G$ is connected implies that edge $u_1u_2$ in $T$ participates in at least one path $P$ of length at most $k$ between two leaves, corresponding to an adjacency in $G$. But because $v_1$ and $v_2$ are represented by the closest leaves to $u_1$ and $u_2$ (respectively) the length of the path in $T$ between the leaves representing $v_1$ and $v_2$ must be at most equal to the length of $P$.
Therefore, there is a path of length at most $k$ between the leaves representing $v_1$ and $v_2$, so $v_1$ and $v_2$ are adjacent in the $k$-leaf power $G$, as required.
\end{proof}

Based on this embedding, we can prove the following characterization of leaf powers,
which we will use in our application of Courcelle's theorem to the problem. It is important, for this characterization, that we express everything intrinsically in terms of the properties of the graph product $G\boxtimes C_k$, its edge coloring, and its subgraphs, without reference to the given graph $G$.
\begin{lemma}
\label{lem:prop}
A given connected graph $G$ on three or more vertices is a $k$-leaf power if and only if the product $G\boxtimes C_k$
has a subgraph $T$ with the following properties:
\begin{enumerate}
\item\label{prop:forest}
$T$ is $1$-degenerate (i.e., a forest).
\item\label{prop:oneleaf}
Every vertex of $G\boxtimes C_k$ is connected by horizontal edges of the product
to exactly one leaf of $T$.
\item\label{prop:adjacent}
Two vertices of $G\boxtimes C_k$ are the endpoints of a non-horizontal edge of the product
if and only if the corresponding leaves of $T$ (given according to Property~\ref{prop:oneleaf})
are the distinct endpoints of a path of length at most $k$ in $T$.
\end{enumerate}
\end{lemma}
\begin{proof}
A subgraph obeying these properties is a forest (Property~\ref{prop:forest}), whose leaves can be placed into one-to-one correspondence with the vertices of~$G$ (Property~\ref{prop:oneleaf}, using the fact that the horizontal cycles of the product correspond one-to-one with vertices of $G$). It has a path of length at most $k$ between two leaves if and only if the corresponding vertices of~$G$ are adjacent (Property~\ref{prop:adjacent}).
So if it exists, it is a $k$-leaf root of $G$ and $G$ is a $k$-leaf power.

In the other direction, if $G$ is a connected $k$-leaf power, let $T$ be a $k$-leaf root of $G$. Then, according to \autoref{lem:leaf-root-embedding}, $T$ can be embedded as a subtree of $G\boxtimes C_k$ (Property~\ref{prop:forest}), with exactly one leaf for each horizontal cycle (Property~\ref{prop:oneleaf}), that forms a $k$-leaf root of $G$ (Property~\ref{prop:adjacent}).
So when $G$ is a $k$-leaf power, a subgraph $T$ obeying the properties of the lemma exists.\end{proof}

\section{Logical Expression}
\label{sec:logic}
In this section, we describe how to express the components of \autoref{lem:prop}, our characterization of the products $G\boxtimes C_k$ that contain a $k$-leaf root of $G$,
in monadic second-order logic. Our logical formula will involve a free variable $horizontal$,
the subset of edges of the given graph (assumed to be of the form $G\boxtimes C_k$) that are horizontal in the product (that is, edges that connect two copies of the same vertex in $G$). 
We will also assume that $V$ and $E$ refer to the vertices and edges of the graph $G\boxtimes C_k$. In our logical formulas, we will express the type of each quantified variable (whether it is a vertex, edge, set of vertices, or set of edges) by annotating its quantifier with a membership or subset relation. For instance, ``$\forall x\in V:\dots$'' quantifies $x$ as a vertex variable.
We will express the incidence predicate between an edge $e$ and a vertex $v$ (true if $v$ is an endpoint of $e$, false otherwise) by $e\incident v$. Because our formulas will also use equality as a predicate, we will express the equality between names of formulas and their explicit logical formulation using a different symbol, $\equiv$.
In our formulas, predicates (equality, incidence, and adjacence) will be considered to bind more tightly than logical connectives, allowing us to omit parentheses in many cases.

A subgraph of the given graph may be represented by its set $S$ of edges. In this representation, adjacency between two vertices $a$ and $b$ may be expressed by the formula
\[
\adjacent(a,b,S)\equiv\exists e\in S: (e\incident a\wedge e\incident b).
\]

The following formula expresses the property that the neighbors of vertex $\ell$ in subgraph $S$ include at most one vertex from a set $X$:
\[
\leaf(\ell,X,S)\equiv
\forall c,d\in X:
\Bigl(\bigl(\adjacent(\ell,c,S)\wedge
\adjacent(\ell,d,S)\bigr)
\rightarrow c=d\Bigr).
\]
This allows us to express the acyclicity of a subgraph $S$ in terms of 1-degeneracy: every nonempty subset $X$ of vertices contains a leaf.
\[
\acyclic(S)\equiv
\forall X\subset V:
(\exists x\in X)\rightarrow\exists \ell\in X: \leaf(\ell,X,S).
\]
This already allows us to express the first condition of  \autoref{lem:prop}.
We will also use a predicate for whether two vertices $p$ and $q$ are connected by horizontal edges. This is true if for every subset $C$ of vertices containing $p$ and excluding $q$, there exists a horizontal edge, connecting a vertex of $C$ to a vertex not in $C$.  
\begin{multline*}
\alignedwith(p,q)\equiv
\forall C\subset V:
\bigl(p\in C\wedge\lnot(q\in C)\bigr)\rightarrow\\
\exists h\in \horizontal:\exists y,z\in V: \
\bigl(y\in C\wedge \lnot(z\in C)\wedge h\incident y\wedge h\incident z\bigr).
\end{multline*}
This allows us to express a predicate for the property that vertex $\ell$ is a leaf of subgraph $S$
on the same horizontal level as another vertex $v$ (that is, $\ell$ is the representative leaf for $v$'s level):
\[
\representative(v,\ell,S)\equiv
\leaf(\ell,V,S)\wedge\alignedwith(v,\ell).
\]
The second part of  \autoref{lem:prop} is that every level has exactly one representative leaf:
\begin{multline*}
\represented(S)\equiv
\bigl(
\forall v\in V: \exists \ell\in V: \representative(v,\ell,S)
\bigr) \wedge\\
\Bigl(
\forall v,\ell_1,\ell_2\in V: \bigl(\representative(v,\ell_1,S)\wedge\representative(v,\ell_2,S)\bigr)\\
\rightarrow \ell_1=\ell_2
\Bigr)
\end{multline*}
{\sc{Unlike}} for the previous formulas, there is no way of expressing the existence of a path of length $k$ from $u$ to $v$ in subgraph $S$, for a non-fixed $k$, in $\MSO_2$. We need a different formula $\haspath_k$ for each $k$. We do not require these paths to be simple, as this would only complicate the formula without simplifying our use of it. However it is essential for our application
to the third condition of \autoref{lem:prop} that we require our paths to have distinct endpoints.
\begin{multline*}
\haspath_k(u,v,S)\equiv
\exists w_1,w_2,\dots w_{k-1}\in V:
\exists e_1,e_2,\dots e_k\in S:\\
\lnot(u=v)\wedge
e_1\incident u \wedge
e_1\incident w_1 \wedge
e_2\incident w_1 \wedge \dots \wedge
e_k\incident w_{k-1} \wedge
e_k\incident v.
\end{multline*}
Other than the inequality of the two endpoints, this formula allows repetitions of vertices and edges within each path. In particular, it allows $w_i$ and $w_{i+1}$ to be equal to each other, repeating one endpoint of an edge twice and omitting the other endpoint. Because we allow repetitions in this way, this formulation of the $\haspath$ predicate has the following convenient property:
\begin{lemma}
For all $k\ge 1$ and all $u$, $v$, and $S$, we have that
\[
\haspath_k(u,v,S)\rightarrow\haspath_{k+1}(u,v,S).
\]
\end{lemma}
\begin{proof}
Let $w_1,\dots w_{k-1}$ and $e_1,\dots e_k$ be the vertices and edges witnessing the truth of $\haspath_k(u,v,S)$, let $w_k=v$, and let $e_{k+1}=e_k$.
Then $w_1,\dots, w_k$ and $e_1,\dots, e_{k+1}$ witness the truth of $\haspath_{k+1}(u,v,S)$.
\end{proof}
\begin{corollary}
Two vertices $u$ and $v$ of a subgraph $S$ of a given graph obey the predicate $\haspath_k(u,v,S)$ if and only if they are distinct and their distance in $S$ is at most $k$.
\end{corollary}
This allows us to express the final part of \autoref{lem:prop}, the requirement that each two vertices are connected by a non-horizontal edge if and only if their representatives are connected by a short path:
\begin{multline*}
\isroot_k(S)\equiv
\forall u,v\in V:
\Bigl(
\bigl(\exists u', v' \in V \ \exists e\in E: \alignedwith(u,u') \wedge \\ \alignedwith(v,v')  \wedge e\incident u'\wedge e\incident v'\wedge\lnot(e\in\horizontal)\bigr) 
\longleftrightarrow \\
\exists x,y\in V:
\bigl(\representative(u,x,S)\wedge
\representative(v,y,S)\wedge \\
\haspath_k(x,y,S)\bigr)
\Bigr).
\end{multline*}
\begin{lemma}
\label{lem:formula}
There exists an $\MSO_2$ formula
that is modeled by a graph $G \boxtimes C_k$
and its set $\horizontal$ of horizontal edges
exactly when $G \boxtimes C_k$ meets the conditions of
\autoref{lem:prop}.
\end{lemma}
\begin{proof}
The formula is
\[
\exists S: \bigl(\acyclic(S)\wedge\represented(S)\wedge\isroot_k(S)\bigr).
\]
A subgraph defined by a set $S$ of its edges meets the first condition of the lemma if $\acyclic(S)$ is true, it meets the second condition of the lemma if $\represented(S)$ is true, and it meets the third condition of the lemma if $\isroot_k(S)$ is true.
\end{proof}
\begin{corollary}
\label{cor:mso}
The property of a graph $G$ being $k$-leaf power can be expressed as an $\MSO_2$
formula of $G\boxtimes C_k$ and of the set $\horizontal$ of horizontal edges of this graph product.
\end{corollary}

\section{Fixed-Parameter Tractability of Leaf Powers}
\label{sec:courcelle}
In this section, by using Courcelle's theorem, we provide our main result that recognizing $k$-leaf powers is fixed-parameter tractable when parameterized by $k$ and the degeneracy of the input graph.  

In order to apply Courcelle's theorem to the graph product $G\boxtimes C_k$ we need to bound its treewidth.
\begin{lemma}
\label{lem:bounded-treewidth}
If $G$ has treewidth $t$ and H has a bounded number of vertices $s$ then $G \boxtimes H$ has treewidth at most $s(t+1)-1$.
\end{lemma}
\begin{proof}
Given any tree-decomposition of $G$ with width $t$, we can form a decomposition of $G\boxtimes H$ by using the same tree, and placing each vertex $(v,w)$ of $G\boxtimes H$ (where $v$ and $w$ are vertices of $G$ and $H$ respectively) into the same bag as vertex $v$ of $G$.
The size of the largest bag of the tree-decomposition of $G$ is $t+1$, so the size of the largest bag of the resulting tree-decomposition of the graph product is $s(t+1)$. The treewidth is one less than the size of the largest bag.
\end{proof} 
\begin{corollary}
If $G$ has a bounded treewidth and $k$ is bounded, then $ G\boxtimes C_k$  also has bounded treewidth.
\label{cor:bounded_tw}
\end{corollary}
This gives us our main theorem:
\begin{theorem}
\label{thm:1}
 For fixed constants $k$ and $d$, it is possible to recognize in linear time (with fixed-parameter tractable dependence on $k$ and $d$) whether a graph of degeneracy at most $d$ is a $k$-leaf power.
 \end{theorem}
 
\begin{proof}
As stated earlier in \autoref{subsec:parameters}, leaf powers with bounded degeneracy have bounded treewidth and it follows from \autoref{cor:bounded_tw} that $G \boxtimes C_k$ also has bounded treewidth. Therefore, by applying Courcelle's theorem to the $\MSO_2$ formula of \autoref{cor:mso} we obtain the result.
\end{proof}

\section{Edges Labeled by Distance Ranges}
\label{sec:generalize}
It is perhaps of interest to generalize $k$-leaf powers to a more general version in which each edge of the input graph $G$ has a weight range $[k_1,k_2]$ where $2 \leq k_1 \leq k_2$ and $K$ is the upper bound on $k_2$ over all the edges. We say that $G$ is a \emph{labeled $K$-leaf power} if $G$ has a $K$-leaf root $T$ in which, for each edge $uv$ of $G$, the corresponding leaves of $T$ are at a distance that is within the range used to label edge $uv$. As with the unlabeled version of the problem, for non-adjacent pairs of vertices of $G$, the corresponding leaves should be at distance more than $K$. The original $k$-leaf power is a restricted variant of this general version in which all edges have a fixed weight range $[1,k]$ and $K = k$. 

One motivation for this comes from the phylogenetic tree applications of $k$-leaf powers.
If we know some information about the evolutionary distance between species, and wish to reconstruct the evolutionary tree, the information we know may be more fine-grained than merely that the distance is big or small. The ranges on each edge allow us to model this fine-grained information and by doing so restrict the trees that can be generated to more accurately reflect the data. As we show in this section, our parameterized algorithms can be extended to the more general problem of recognizing labeled $K$-leaf powers.

Recall that we are already modeling some labeling information on the graph product $G_1 \boxtimes G_2$, in the logic of graphs, as the free set variable $\horizontal$.
We will similarly need to model the edge weight range labels logically.
To do so, we extend the weights on the edges of $G$ to the weights on the edges of a graph product using the following definition.
Suppose that we are considering the graph product $G_1 \boxtimes G_2$ where $G_1$ and $G_2$ are weighted and unweighted, respectively. Recall that, in this product, two distinct vertices $(u_1,u_2)$ and $(v_1,v_2)$ are adjacent if and only if for all $i \in \{1,2\}$,  $u_i = v_i$ or $u_i$ and $v_i$ are adjacent in $G_i$. A vertical or diagonal edge is an edge with endpoints  $(u_1,u_2)$ and $(v_1,v_2)$, for which $u_1\ne v_1$. In this case, we assign the vertical or diagonal edge weight $\omega$ if  the edge connecting $u_1$ and $v_1$ has weight $\omega$, in $G_1$.

We have the following analogue of \autoref{lem:leaf-root-embedding} for the weighted case:

\begin{lemma}
\label{lem:general-leaf-root-embedding}
If $G$ is a weighted connected $K$-leaf power graph on three or more vertices, and $T$ is any $K$-leaf root of $G$, then $T$
can be embedded as a subtree of the strong product $G  \boxtimes C_{K}$.  
Additionally, the embedding can be chosen in such a way that each horizontal cycle in the strong product (the product of a vertex $v$ of $G$ with $C_{K}$) contains exactly one leaf of the embedded copy of $T$, the leaf representing $v$.
\end{lemma}

\begin{proof}
The weighted graph product has the same underlying graph as the unweighted product,
and the weighted $K$-leaf root is a special case of the unweighted $K$-leaf root,
so this follows immediately from \autoref{lem:leaf-root-embedding}, which provides an embedding into the graph power of every $K$-leaf root.
\end{proof}

We can now provide the following characterization of $K$-leaf powers.

\begin{lemma}
\label{lem:general-prop}
A given connected weighted graph $G$ on three or more vertices is a $K$-leaf power if and only if the product $G\boxtimes C_{K}$
has a subgraph $T$ with the following properties:
\begin{enumerate}
\item\label{general-prop:forest}
$T$ is $1$-degenerate (i.e., a forest).
\item\label{general-prop:oneleaf}
Every vertex of $G\boxtimes C_{K}$ is connected by horizontal edges of the product to exactly one leaf of $T$.
\item\label{general-prop:adjacent}
If two vertices of $G\boxtimes C_{K}$ are the endpoints of a non-horizontal edge of the product with weight $[k_1,k_2]$
then the corresponding leaves of $T$ (given according to Property~\ref{general-prop:oneleaf})
are the distinct endpoints of a path of length at least $k_1$ and at most $k_2$ in $T$.
\item\label{general-prop:nonadjacent}
If two distinct leaves of $T$ are at distance at most $K$ then there exists a non-horizontal edge of the product with two endpoints vertices, aligned to each leaf.
\end{enumerate}
\end{lemma}

\begin{proof}
The proof follows the same lines as the proof of  \autoref{lem:prop}, modified only to take into account the edge weights.
\end{proof}

In order to express the components of  \autoref{lem:general-prop} in monadic second-order logic, we reuse formulas \acyclic\ and \represented \ from  \autoref{lem:prop} for the first and second parts of   \autoref{lem:general-prop}, respectively. 

To express the third part, we introduce $K^2$ edges sets $I_{k_1,k_2}$ where $2 \leq k_1 \leq k_2 \leq K$. An edge $e$ of the product, with two endpoints $(u_1,u_2)$ and $(v_1,v_2)$ belongs to $I_{k_1,k_2}$ if and only if $u_1 \ne u_2$, $v_1 \ne v_2$ and it has weight $[k_1,k_2]$. This allows us the express the requirement that if two vertices are connected by a non-horizontal edge with weight $[k_1,k_2]$ then their representatives are connected by a path with a length in the range $[k_1,k_2]$:

\begin{multline*}
\edge_{k_1,k_2}(S)\equiv
\forall u,v\in V:
\Bigl(
\bigl(\exists e\in E: e\incident u\wedge e\incident v\wedge (e\in I_{k_1,k_2})\bigr) \\
\longrightarrow
\exists x,y\in V:
\bigl(\representative(u,x,S)\wedge
\representative(v,y,S)\wedge \\
\haspath_{k_2}(x,y,S) \wedge \neg \haspath_{k_1-1}(x,y,S) \bigr)
\Bigr).
\end{multline*}

The last part of  \autoref{lem:general-prop} can be expressed as follows:

\begin{multline*}
\nonedge_{K}(S)\equiv
\forall u,v\in V:
\Bigl( \bigl( \exists x,y \in V: \representative(u,x,S) \wedge \\ \representative(y,v,S)\wedge \haspath_{K}(x,y,S) \bigr) \longrightarrow \bigl(\exists u',v' \in V \  \exists e \in E: e\incident u' \\ \wedge e\incident v' \wedge \alignedwith(u,u') \wedge  \alignedwith(v,v') \wedge \neg (e \in \horizontal)\bigr) \Bigr)
\end{multline*}

\begin{lemma}
\label{lem:general-formula}
There exists an $\MSO_2$ formula
that is modeled by a graph $G \boxtimes C_{K}$
and its set $\horizontal$ of horizontal edges and $K^2$ edge sets $I_{k_1,k_2}$
exactly when $G \boxtimes C_{K}$ meets the conditions of
\autoref{lem:general-prop}.
\end{lemma}
\begin{proof}
The formula is
\begin{multline*}
\exists S: \bigl(\acyclic(S)\wedge\represented(S)\wedge\edge_{2,2}(S) \wedge \edge_{2,3}(S) \wedge  \dots  \wedge \\ \edge_{K,K}(S)  \wedge \nonedge_{K}(S)\bigr).
\end{multline*}
A subgraph defined by a set $S$ of its edges meets the first condition of the  \autoref{lem:general-prop} if $\acyclic(S)$ is true, it meets the second condition of the lemma if $\represented(S)$ is true, it meets the third condition of the lemma if $\edge_{k_1,k_2}(S)$ is true for all $2 \leq k_1 \leq k_2 \leq K$, and it meets the forth condition of the lemma if $\nonedge(S)$ is true.
\end{proof}
\begin{corollary}
\label{cor:general-mso}
The property of a weighted graph $G$ being $K$-leaf power can be expressed as an $\MSO_2$
formula of $G\boxtimes C_{K}$, of the set $\horizontal$ of horizontal edges and of the $K^2$ edge sets $I_{k_1,k_2}$ of this graph product.
\end{corollary}

As proved in \autoref{lem:bounded-treewidth}, if $G$ has a bounded treewidth and $K$ is fixed, then $G \boxtimes C_{K}$ also has a bounded treewidth. This fact enables us to provide the following theorem for the general leaf power problem.
\begin{theorem}
 For fixed constants $K$ and $d$, it is possible to recognize in linear time (with fixed-parameter tractable dependence on $K$ and $d$) whether a graph of degeneracy at most $d$ is a $K$-leaf power.
 \end{theorem}
 
\begin{proof}
The proof follows the same outline as the proof of  \autoref{thm:1}, modified only to use the weighted versions of the lemmas above in place of their unweighted versions.
\end{proof}

\section{Dynamic Programming Algorithm}
\label{sec:dynamic}

Many graph problems, including a vast number of NP-hard problems, have been shown to be solvable in polynomial time when given a tree decomposition of constant width~\cite{MR1105479,bodlaender1988dynamic,MR1268488}. 
Dynamic programming on tree decomposition of graphs is an underlying technique to devise such algorithms, restricted to graphs of bounded treewidth~\cite{bodlaender1988dynamic}. Indeed, our application of Courcelle's theorem relies on such an algorithm to evaluate whether a logical formula is modeled by the given graph.  In this section, we present a direct dynamic programming algorithm to decide whether the input graph is a $k$-leaf power. 

Dynamic programming algorithms often use a variant of tree decomposition, called \textit{nice} tree decomposition. A nice tree decomposition of graph $G$ is a rooted tree decomposition $T$ of $G$ in which each bag $X_i$ is one of the following:
\begin{itemize}
\item a \textit{leaf} bag in which $|X_i|$ = 1,
\item a \textit{forget} bag with one child $X_j$, where $X_i \subset X_j$ and $|X_j| - |X_i| = 1$,
\item an \textit{introduce} bag with one child $X_j$, where $X_j \subset X_i$ and $|X_i| - |X_j| = 1$, or
\item a \textit{join} bag with two children $X_j$ and $X_{j'}$, where $X_i = X_j = X_{j'}$,
\end{itemize}

For a forget bag we call  $X_j \setminus X_i$ the \textit{forgotten vertex}. Given a graph $G$ and its tree decomposition of width $w$, one can construct a nice tree decomposition of equal width in linear time~\cite{kloks1994treewidth}.  
Our algorithm uses these restrictions on tree decompositions, but we need others as well.
Therefore, we will define an \textit{extra} nice tree decomposition.
In comparison with nice tree decomposition, an extra nice tree decomposition has one more type of bag, an \textit{edge-associated}  bag. An edge-associated bag $X_i$ has a child $X_j$ where  $X_i = X_j$ and exactly one edge $e(u,v)$, $u,v \in X_i$, is associated with $X_i$. Using a nice tree decomposition of $G$, we can simply construct such tree decomposition in the following way: for each pair of adjacent pairs $u$ and $v$ in bag $X_i$, if $e(u,v)$ is not yet associated to a bag, create a new bag $X_{i'}$ as a new parent of $X_i$ where $X_{i'} = X_i$ and associate edge $e(u,v)$ to $X_{i'}$. The old parent of $X_i$, if it exists, is now the parent of $X_{i'}$.

Our algorithm is run over a \textit{mixed} decomposition of graphs $G$ and graph product $H$. Given an extra nice decomposition of $G$ of width $w$, for each vertex $v$ in bag $X_i$, add all vertices $(v,r) \in H$ for $ 0 \leq r < k$. Hence, the size of each bag of the mixed  decomposition is at most $wk$.  Our second algorithm can therefore be viewed as using the same graph product technique that our first algorithm used, applied directly in a dynamic programming algorithm  rather than indirectly via Courcelle's theorem.

\subsection{Local Picture of a $\pmb k$-leaf root}

Intuitively, for each bag of mixed decomposition $M$, we describe, a local picture which describes a subtree of a $k$-leaf root $T$, if one exists. This  description allows us to check whether the big picture, $T$, is a $k$-leaf root of $G$. For a bag $X_i$ let $G_i$ and $H_i$ be a set of vertices of  $X_i$ that belongs to $G$ and $H$, respectively. 

A local picture of $T$ at bag $X_i$ consists of the following ingredients:
\begin {itemize}
\item A partition of $H_i$ into connected components (with one more partition set for vertices of $H_i$, not participating in $T$).
\item A distance matrix between each pair of vertices in the same component. Each coefficient of the matrix will store either a number between 1 and $k$ (the distance between two vertices), or a special flag $\infty$ to represent a finite distance greater than $k$.
\item A designated root vertex for each component, the vertex that will become the closest to the root of $T$. 
\item For each vertex $v$ of $G_i$ in $X_i$, a corresponding vertex $(v,i)$ chosen as the leaf representative of $v$ in $H_i$.
\end{itemize}
To reduce the number of local pictures that we need to consider,
consistently with the embedding of \autoref{sec:embedding},
we will restrict our attention to local pictures in which the vertices $(v,i)$ of $H_i$ associated with a single vertex $v$ of $H_i$ are either part of a single component or not in any component,
and have distances within that component consistent with their distances along the cycle $C_k$.
We will associate with each remaining local picture a Boolean variable. We will set this variable to True if there is a subtree of $H$ within the bags descending from $X_i$ that is consistent with the local picture
and with the requirement that it be part of a leaf root of $G$. Otherwise, we set this variable to False.   
In order to enforce the requirement that the local picture be consistent with being part of a leaf root, we only consider local pictures such that, for the distances in
each component, the pairs of representative vertices at distance at
most $k$ are adjacent in $G$ and pairs with distance $\infty$ are non-adjacent.
Adjacent vertices in $G_i$ whose representatives belong to different components are allowed, however, as their distance will be checked at a higher level of the tree decomposition where their components merge. If these conditions are not met, we set the associated Boolean variable of the local picture to False. 

We process $M$ in post-order from leaves to the root of $M$ computing for each bag and each local picture the Boolean variable for that local picture. This bottom-up ordering ensures that the variables for local pictures of the child or children of a bag are known before we try to compute the variables at the bag itself.  After computing these values, $G$ will be a $k$-leaf power if and only if there exists a local picture at the root bag whose associated Boolean variable is true. If $G$ is a $k$-leaf power, one can form a $k$-leaf root by creating a vertex as the root of the $k$-leaf root and connect it to the root of each component of the True local picture, with an appropriate number of edges (at most k edges for each connection). Further, such ordering allows us to remember the distance to the nearest forgotten leaf as $\mu_v$  for each non-leaf vertex $v$ of each local picture for distance-checking purposing. In another word, $\mu_v$ stores the distance from $v$ to the nearest forgotten leaf that is no more present in the current local picture. When the bottom-up traversal of $M$ reaches a bag $X_i$, one of the following cases occurs: 

\begin{itemize}
\item $X_i$ may be a leaf of $M$. In this case, it contains a vertex $v \in G$ alongside all vertices $(v,r)$, $ 0 \leq r < k$. A local picture is set to True if and only if it has one component, a single chain of vertices with the appropriate distances, ending at the vertex designated as the  representative of $v$.

\item $X_i$ may be a forget bag. In this case, it has one child $X_j$ where $X_i \subset X_j$ and $H_j \wedge (X_j \backslash X_i) = \{(v,r)\}, 0 \leq r < k$. A local picture $\ell$ at $X_i$ is set to True if and only if it is formed by removing vertices $(v,r)$ (a chain of vertices representing $v \in G_j$) from a True local picture $\ell'$ of $X_j$.The removal of such chain of vertices may result in more number of components in the corresponding True local picture $\ell$. If a removed vertex has a child other than the one in the chain, that child becomes the root of a new component in $\ell$. Further, as the designated leaf of such chain is forgotten, there might be a need to update $\mu_u$ for a vertex $u$ in $\ell$ within the vicinity ($<k$) of the forgotten leaf. 

\item $X_i$ may be an introduce bag. In this case, it has one child $X_j$ where $X_j \subset X_i$ and $G_i \wedge (X_i \backslash X_j) = \{v\}$. A local picture at $X_i$ is set to True if and only if it can be formed from one of the True local pictures of $X_j$ by adding one more component which is a path $(v,r), \dots , (v,r')$, $ 0 \leq r,r' < k$. Because the subtree descending from $X_i$ does not contain any edge-associated bags for edges incident with $v$, this component cannot be connected to any of the existing components in the local picture in $X_j$.

\item $X_i$ may be an edge-associated bag. In this case, it has one child $X_j$ where $X_i = X_j$ and there exists an edge $e(u,v)$ associated to bag $X_i$. A local picture $L$ at $X_i$ is set True if and only if either there exists an exact True copy of the local picture at $X_j$, or using the edge $e(u,v)$, $L$ can be formed from a True local picture at $X_j$ by connecting a root $x$ of one component to a vertex $w$ of another component.  Such connection can be made if the resulting local picture obeys the distance matrix and also the distance from each forgotten leaf of one component to a (forgotten or existing) leaf of another component is greater than $k$ as their corresponding vertices  in $G$  cannot be adjacent given the definition of extra nice decomposition (when a vertex is forgotten, it cannot be reintroduced as the bags containing that vertex form a nonempty connected subtree). 

\item $X_i$ may be a join bag. In this case, it has two children $X_j$ and $X_{j'}$ where $X_i= X_j = X_{j'}$. A local picture $L_1$ at $X_i$ has its value set to True if and only if there exist True  local pictures $L_2$ and $L'_2$ at $X_j$ and $X_{j'}$, respectively, that when combined together, they form $L_1$. To find such a combination, we consider all pairs of local pictures for $L_2$ and $L'_2$  at $X_j$ and $X_{j'}$ and construct a bipartite graph $F$. One side of bipartition includes vertices of $H_i \in X_i$, each with two neighbors, representing the two subtrees, the vertex belongs to in the local pictures $L_2$ and $L'_2$. $L_1$ can be formed if and only if $F$ is a forest, its subtrees are subtrees of $F$ and the combined local picture obeys the distance matrix at $L_1$ and no forgotten or existing leaf of $L_2$ get a distance at most $k$ to a forgotten leaf of $L_2'$ or vice versa. 

\end{itemize}

\subsection{Analysis}

To analyze our dynamic programming algorithm, we need to understand the number of local pictures that are possible in each bag of the tree decomposition. We can perform this analysis by combining the following factors, each of which depends only on the width $w$ and leaf power parameter $k$ of the given input.
\begin{itemize}
\item For each vertex $v$ of $G_i$, there are $O(k^2)$ choices for the representative vertex and the length of the path using vertices $(v,i)$ in the component of this representative vertex. The total number of such choices for all vertices of $G_i$ is $k^{O(w)}$.
\item Given these choices of paths, there are $w^{O(w)}$ ways of connecting the paths into components and selecting the vertex closest to the root within each component.
\item Within a component that connects $c$ paths, there are $(ck)^{O(w)}$ choices of distance matrix for the whole component consistent with the distances within each path and with the assumption that the distances come from a tree.
\end{itemize}

Therefore, there are $(wk)^{O(w)}$ local pictures considered by our algorithm for each bag.
The time for the algorithm is dominated by the join bags; there are $n-1$ of these bags, and in each such bag we consider a number of pairs of local pictures bounded by the square of the number of local pictures per bag. Each pair of local pictures in the two child bags takes time polynomial in $w$ and $k$ to check for whether it is consistent and to find the corresponding local picture in the join bag. So the total time for our dynamic programming algorithm is $O\bigl(n(wk)^{O(w)}\bigr)$.

\section{Conclusion}
\label{sec:conclusion}

We have provided two fixed-parameter algorithms to recognize $k$-leaf powers (and generalized $K$-leaf powers) for graphs of bounded degeneracy. In both methods we use embeding of a $k$-leaf root of a $k$-leaf power graph in the graph product of the input graph and a $k$-vertex cycle $C_k$. Our first algorithm finds a logical characterization of the leaf roots that are embedded in this way, and applies Courcelle's theorem to determine the existence of a subgraph of the graph product that meets our characterization.


Our methods of using low-treewidth supergraphs to represent vertices and edges that are not part of the input graph, and of using graph products to find these supergraphs helped us to solve the problem directly using dynamic programming rather than by applying Courcelle's theorem. Additionally, these methods may be useful in other graph problems. For instance, the same graph product technique would have greatly simplified the application of Courcelle's theorem in our recent work on planar split thickness~\cite{MR3761167}: a graph $G$ has planar split thickness $k$ if and only if $G\boxtimes K_k$ has a planar subgraph $S$ such that, for each non-horizontal edge of the product, the endpoints of the edge are aligned with the endpoints of an edge in $S$. In reducing the logical complexity of problems such as these, our first method also makes it more likely that faster model checkers for restricted
fragments of MSO logic \cite{bannach_et_al:LIPIcs:2018:9469} can be applied to our problem. 

Our dynamic programming algorithm has significantly better dependence on its parameters than
our first, logic-based algorithm. However, its dependence is still not singly exponential.
We leave whether this is possible as open for future research.

\bibliography{ref}   

\begin{thebibliography}{10}

\bibitem{alon2009linear}
Noga Alon and Shai Gutner.
\newblock Linear time algorithms for finding a dominating set of fixed size in
  degenerated graphs.
\newblock {\em Algorithmica}, 54(4):544, 2009.

\bibitem{MR1105479}
Stefan Arnborg, Jens Lagergren, and Detlef Seese.
\newblock {Easy problems for tree-decomposable graphs}.
\newblock {\em J. Algorithms}, 12(2):308{--}340, 1991.
\newblock \href {http://dx.doi.org/10.1016/0196-6774(91)90006-K}
  {\path{doi:10.1016/0196-6774(91)90006-K}}.

\bibitem{bannach_et_al:LIPIcs:2018:9469}
Max Bannach and Sebastian Berndt.
\newblock {Practical Access to Dynamic Programming on Tree Decompositions}.
\newblock In Yossi Azar, Hannah Bast, and Grzegorz Herman, editors, {\em 26th
  Annual European Symposium on Algorithms (ESA 2018)}, volume 112 of {\em
  Leibniz International Proceedings in Informatics (LIPIcs)}, pages 6:1--6:13,
  Dagstuhl, Germany, 2018. Schloss Dagstuhl--Leibniz-Zentrum fuer Informatik.
\newblock URL: \url{http://drops.dagstuhl.de/opus/volltexte/2018/9469}, \href
  {http://dx.doi.org/10.4230/LIPIcs.ESA.2018.6}
  {\path{doi:10.4230/LIPIcs.ESA.2018.6}}.

\bibitem{bannister2014crossing}
Michael~J Bannister and David Eppstein.
\newblock Crossing minimization for 1-page and 2-page drawings of graphs with
  bounded treewidth.
\newblock In {\em International Symposium on Graph Drawing}, pages 210--221.
  Springer, 2014.

\bibitem{BB72}
Umberto Bertel{\'e} and Francesco Brioschi.
\newblock {\em {Nonserial Dynamic Programming}}.
\newblock Academic Press, 1972.

\bibitem{bodlaender1988dynamic}
Hans~L Bodlaender.
\newblock Dynamic programming on graphs with bounded treewidth.
\newblock In {\em International Colloquium on Automata, Languages, and
  Programming}, pages 105--118. Springer, 1988.

\bibitem{MR1268488}
Hans~L. Bodlaender.
\newblock {A tourist guide through treewidth}.
\newblock {\em Acta Cybernet.}, 11(1-2):1{--}21, 1993.

\bibitem{brandstadt2008ptolemaic}
Andreas Brandst{\"a}dt and Christian Hundt.
\newblock Ptolemaic graphs and interval graphs are leaf powers.
\newblock In {\em Latin American Symposium on Theoretical Informatics}, pages
  479--491. Springer, 2008.

\bibitem{MR2574841}
Andreas Brandst{\"a}dt, Christian Hundt, Federico Mancini, and Peter Wagner.
\newblock {Rooted directed path graphs are leaf powers}.
\newblock {\em Discrete Math.}, 310(4):897{--}910, 2010.
\newblock \href {http://dx.doi.org/10.1016/j.disc.2009.10.006}
  {\path{doi:10.1016/j.disc.2009.10.006}}.

\bibitem{MR2211095}
Andreas Brandst{\"a}dt and Van~Bang Le.
\newblock {Structure and linear time recognition of 3-leaf powers}.
\newblock {\em Inform. Process. Lett.}, 98(4):133{--}138, 2006.
\newblock \href {http://dx.doi.org/10.1016/j.ipl.2006.01.004}
  {\path{doi:10.1016/j.ipl.2006.01.004}}.

\bibitem{MR2537378}
Andreas Brandst{\"a}dt, Van~Bang Le, and Dieter Rautenbach.
\newblock {A forbidden induced subgraph characterization of distance-hereditary
  5-leaf powers}.
\newblock {\em Discrete Math.}, 309(12):3843{--}3852, 2009.
\newblock \href {http://dx.doi.org/10.1016/j.disc.2008.10.025}
  {\path{doi:10.1016/j.disc.2008.10.025}}.

\bibitem{MR2479182}
Andreas Brandst{\"a}dt, Van~Bang Le, and R.~Sritharan.
\newblock {Structure and linear-time recognition of 4-leaf powers}.
\newblock {\em ACM Trans. Algorithms}, 5(1):A11:1{--}A11:22, 2009.
\newblock \href {http://dx.doi.org/10.1145/1435375.1435386}
  {\path{doi:10.1145/1435375.1435386}}.

\bibitem{brandstadt2008k}
Andreas Brandst{\"a}dt and Peter Wagner.
\newblock On k-versus (k+ 1)-leaf powers.
\newblock In {\em International Conference on Combinatorial Optimization and
  Applications}, pages 171--179. Springer, 2008.

\bibitem{cai2006random}
Leizhen Cai, Siu~Man Chan, and Siu~On Chan.
\newblock Random separation: A new method for solving fixed-cardinality
  optimization problems.
\newblock In {\em International Workshop on Parameterized and Exact
  Computation}, pages 239--250. Springer, 2006.

\bibitem{chang20073}
Maw-Shang Chang and Ming-Tat Ko.
\newblock The 3-steiner root problem.
\newblock In {\em International Workshop on Graph-Theoretic Concepts in
  Computer Science}, pages 109--120. Springer, 2007.

\bibitem{chang2015linear}
Maw-Shang Chang, Ming-Tat Ko, and Hsueh-I Lu.
\newblock Linear-time algorithms for tree root problems.
\newblock {\em Algorithmica}, 71(2):471--495, 2015.

\bibitem{MR2001887}
Zhi-Zhong Chen, Tao Jiang, and Guohui Lin.
\newblock {Computing phylogenetic roots with bounded degrees and errors}.
\newblock {\em SIAM J. Comput.}, 32(4):864{--}879, 2003.
\newblock \href {http://dx.doi.org/10.1137/S0097539701389154}
  {\path{doi:10.1137/S0097539701389154}}.

\bibitem{MR1042649}
Bruno Courcelle.
\newblock {The monadic second-order logic of graphs. I. Recognizable sets of
  finite graphs}.
\newblock {\em Inform. and Comput.}, 85(1):12{--}75, 1990.
\newblock \href {http://dx.doi.org/10.1016/0890-5401(90)90043-H}
  {\path{doi:10.1016/0890-5401(90)90043-H}}.

\bibitem{MR1451381}
Bruno Courcelle.
\newblock {On the expression of graph properties in some fragments of monadic
  second-order logic}.
\newblock In Neil Immerman and Phokion~G. Kolaitis, editors, {\em Descriptive
  Complexity and Finite Models: Proceedings of a DIMACS Workshop, January
  14{--}17, 1996, Princeton University}, volume~31 of {\em DIMACS Ser. Discrete
  Math. Theoret. Comput. Sci.}, pages 33{--}62. American Mathematical Society,
  Providence, RI, 1997.

\bibitem{MR1480957}
Bruno Courcelle.
\newblock {The expression of graph properties and graph transformations in
  monadic second-order logic}.
\newblock In {\em Handbook of graph grammars and computing by graph
  transformation, Vol. 1}, pages 313{--}400. World Scientific, River Edge, NJ,
  1997.
\newblock \href {http://dx.doi.org/10.1142/9789812384720_0005}
  {\path{doi:10.1142/9789812384720_0005}}.

\bibitem{MR1217156}
Bruno Courcelle, Joost Engelfriet, and Grzegorz Rozenberg.
\newblock {Handle-rewriting hypergraph grammars}.
\newblock {\em J. Comput. System Sci.}, 46(2):218{--}270, 1993.
\newblock \href {http://dx.doi.org/10.1016/0022-0000(93)90004-G}
  {\path{doi:10.1016/0022-0000(93)90004-G}}.

\bibitem{MR1739644}
Bruno Courcelle, J.~A. Makowsky, and U.~Rotics.
\newblock {Linear time solvable optimization problems on graphs of bounded
  clique-width}.
\newblock {\em Theory Comput. Syst.}, 33(2):125{--}150, 2000.
\newblock \href {http://dx.doi.org/10.1007/s002249910009}
  {\path{doi:10.1007/s002249910009}}.

\bibitem{dom2004error}
Michael Dom, Jiong Guo, Falk H{\"u}ffner, and Rolf Niedermeier.
\newblock Error compensation in leaf root problems.
\newblock In {\em International Symposium on Algorithms and Computation}, pages
  389--401. Springer, 2004.

\bibitem{dom2005extending}
Michael Dom, Jiong Guo, Falk H{\"u}ffner, and Rolf Niedermeier.
\newblock Extending the tractability border for closest leaf powers.
\newblock In {\em International Workshop on Graph-Theoretic Concepts in
  Computer Science}, pages 397--408. Springer, 2005.

\bibitem{ducoffe20194}
Guillaume Ducoffe.
\newblock The 4-steiner root problem.
\newblock In {\em International Workshop on Graph-Theoretic Concepts in
  Computer Science}, pages 14--26. Springer, 2019.

\bibitem{ducoffe2019finding}
Guillaume Ducoffe.
\newblock Finding cut-vertices in the square roots of a graph.
\newblock {\em Discrete Applied Mathematics}, 257:158--174, 2019.

\bibitem{MR3761167}
David Eppstein, Philipp Kindermann, Stephen Kobourov, Giuseppe Liotta, Anna
  Lubiw, Aude Maignan, Debajyoti Mondal, Hamideh Vosoughpour, Sue Whitesides,
  and Stephen Wismath.
\newblock {On the planar split thickness of graphs}.
\newblock {\em Algorithmica}, 80(3):977{--}994, 2018.
\newblock \href {http://dx.doi.org/10.1007/s00453-017-0328-y}
  {\path{doi:10.1007/s00453-017-0328-y}}.

\bibitem{eppstein2010listing}
David Eppstein, Maarten L{\"o}ffler, and Darren Strash.
\newblock Listing all maximal cliques in sparse graphs in near-optimal time.
\newblock In {\em International Symposium on Algorithms and Computation}, pages
  403--414. Springer, 2010.

\bibitem{fitch1967construction}
Walter~M. Fitch and Emanuel Margoliash.
\newblock {Construction of phylogenetic trees}.
\newblock {\em Science}, 155(3760):279{--}284, 1967.
\newblock \href {http://dx.doi.org/10.1126/science.155.3760.279}
  {\path{doi:10.1126/science.155.3760.279}}.

\bibitem{golovach2016finding}
Petr~A Golovach, Dieter Kratsch, Dani{\"e}l Paulusma, and Anthony Stewart.
\newblock Finding cactus roots in polynomial time.
\newblock In {\em International Workshop on Combinatorial Algorithms}, pages
  361--372. Springer, 2016.

\bibitem{MR2120320}
Martin Grohe.
\newblock {Computing crossing numbers in quadratic time}.
\newblock In {\em Proceedings of the Thirty-Third Annual ACM Symposium on
  Theory of Computing}, pages 231{--}236, New York, 2001. ACM.
\newblock \href {http://dx.doi.org/10.1145/380752.380805}
  {\path{doi:10.1145/380752.380805}}.

\bibitem{gurski2007clique}
Frank Gurski and Egon Wanke.
\newblock The clique-width of tree-power and leaf-power graphs.
\newblock In {\em International Workshop on Graph-Theoretic Concepts in
  Computer Science}, pages 76--85. Springer, 2007.

\bibitem{MR0444522}
Rudolf Halin.
\newblock {$S$-functions for graphs}.
\newblock {\em J. Geometry}, 8(1-2):171{--}186, 1976.
\newblock \href {http://dx.doi.org/10.1007/BF01917434}
  {\path{doi:10.1007/BF01917434}}.

\bibitem{MR2220663}
Petr Hlin{\v{e}}n{\'y}.
\newblock {Branch-width, parse trees, and monadic second-order logic for
  matroids}.
\newblock {\em J. Combin. Theory Ser. B}, 96(3):325{--}351, 2006.
\newblock \href {http://dx.doi.org/10.1016/j.jctb.2005.08.005}
  {\path{doi:10.1016/j.jctb.2005.08.005}}.

\bibitem{MR2577678}
William Kennedy, Guohui Lin, and Guiying Yan.
\newblock {Strictly chordal graphs are leaf powers}.
\newblock {\em J. Discrete Algorithms}, 4(4):511{--}525, 2006.
\newblock \href {http://dx.doi.org/10.1016/j.jda.2005.06.005}
  {\path{doi:10.1016/j.jda.2005.06.005}}.

\bibitem{kloks1994treewidth}
Ton Kloks.
\newblock {\em Treewidth: computations and approximations}, volume 842.
\newblock Springer Science \& Business Media, 1994.

\bibitem{lau2006bipartite}
Lap~Chi Lau.
\newblock Bipartite roots of graphs.
\newblock {\em ACM Transactions on Algorithms (TALG)}, 2(2):178--208, 2006.

\bibitem{lick1970k}
Don~R Lick and Arthur~T White.
\newblock k-degenerate graphs.
\newblock {\em Canadian J. of Mathematics}, 22:1082--1096, 1970.

\bibitem{MR709826}
David~W. Matula and Leland~L. Beck.
\newblock {Smallest-last ordering and clustering and graph coloring
  algorithms}.
\newblock {\em J. ACM}, 30(3):417{--}427, 1983.
\newblock \href {http://dx.doi.org/10.1145/2402.322385}
  {\path{doi:10.1145/2402.322385}}.

\bibitem{nguyen2009hardness}
Ngoc~Tuy Nguyen et~al.
\newblock Hardness results and efficient algorithms for graph powers.
\newblock In {\em International Workshop on Graph-Theoretic Concepts in
  Computer Science}, pages 238--249. Springer, 2009.

\bibitem{MR1874637}
Naomi Nishimura, Prabhakar Ragde, and Dimitrios~M. Thilikos.
\newblock {On graph powers for leaf-labeled trees}.
\newblock {\em J. Algorithms}, 42(1):69{--}108, 2002.
\newblock \href {http://dx.doi.org/10.1006/jagm.2001.1195}
  {\path{doi:10.1006/jagm.2001.1195}}.

\bibitem{MR2237730}
Dieter Rautenbach.
\newblock {Some remarks about leaf roots}.
\newblock {\em Discrete Math.}, 306(13):1456{--}1461, 2006.
\newblock \href {http://dx.doi.org/10.1016/j.disc.2006.03.030}
  {\path{doi:10.1016/j.disc.2006.03.030}}.

\bibitem{MR855559}
Neil Robertson and P.~D. Seymour.
\newblock {Graph minors. II. Algorithmic aspects of tree-width}.
\newblock {\em J. Algorithms}, 7(3):309{--}322, 1986.
\newblock \href {http://dx.doi.org/10.1016/0196-6774(86)90023-4}
  {\path{doi:10.1016/0196-6774(86)90023-4}}.

\bibitem{tuy2010square}
Nguyen~Ngoc Tuy et~al.
\newblock The square of a block graph.
\newblock {\em Discrete Mathematics}, 310(4):734--741, 2010.

\end{thebibliography}

%
%

\end{document}